\documentclass[12pt]{amsart}
\usepackage[utf8]{inputenc}
\usepackage{xcolor}
\usepackage[margin=1.5in]{geometry}
\usepackage{indentfirst}
\usepackage{amssymb}

\theoremstyle{plain}
\newtheorem{thm}{Theorem} 
\newtheorem{lem}[thm]{Lemma}
\newtheorem{clm}{Claim}

\theoremstyle{definition}

\title{Complexity Results for Implication Bases of Convex Geometries}
\author{Todd Bichoupan}
\address{Department of Computer Science, Hofstra University, Hempstead, NY  11549}
\email{tbichoupan1@pride.hofstra.edu}
\date{November 15, 2022}

\begin{document}

\maketitle

\begin{abstract}
\noindent
A convex geometry is finite zero-closed closure system that satisfies the anti-exchange property. Complexity results are given for two open problems related to representations of convex geometries using implication bases. In particular, the problem of optimizing an implication basis for a convex geometry is shown to be NP-hard by establishing a reduction from the minimum cardinality generator problem for general closure systems. Furthermore, even the problem of deciding whether an implication basis defines a convex geometry is shown to be co-NP-complete by a reduction from the Boolean tautology problem.
\end{abstract}

\section{Introduction}
Convex geometries are finite zero-closed closure systems satisfying the anti-exchange property. P. Edelman and R. Jamison introduced convex geometries in \cite{EdJa95} as a common formulation for various equivalent structures appearing in several areas of research, including lattice theory and relational databases. More recently, K. Adaricheva and J.B. Nation published a survey of work related to convex geometries \cite{AdNa16}. H. Yoshikawa et al. found a notable application for convex geometries in computer-driven education systems, where convex geometries are used in their dual form as \emph{antimatroids} to model the space of possible knowledge states of learners \cite{Yosh15}. 

This paper continues an effort to study convex geometries through the lens of implication bases.

A set of implications $\Sigma = \{A_i \rightarrow B_i : 1 \leq i \leq n\}$ defining a finite closure system is an \emph{optimum basis} if its size is minimal among all implication bases for the same closure system; the size of an implication basis is the sum of all cardinalities of the sets on the left and right sides of its implications. D. Maier showed that the general problem of finding an optimum basis for a closure system given by an arbitrary implication basis is NP-hard \cite{Mai80}.

On the other hand, several special classes of closure systems are known to have tractable optimum bases. M. Wild showed tractability of optimum bases for modular closure systems in \cite{Wild00}, and E. Boros et al. showed tractability of optimum bases for component-quadratic closure systems in \cite{Bor09}. In particular, many subclasses of the convex geometries are known to have tractable optimum bases. K. Kashiwabara and M. Nakamura proved tractability of optimum bases for affine geometries \cite{KaNa13}, and K. Adaricheva showed tractability of optimum bases for convex geometries that satisfy the Carousel property, convex geometries that do not have $D$-cycles, and convex geometries of order convex subsets in a poset \cite{Ada17}. But the problem of determining whether there is a tractable algorithm to find optimum bases of convex geometries in general has remained open.

We show that, given an arbitrary implication basis for a convex geometry, the problem of finding an optimum basis is NP-hard. The proof relies on a reduction from the \emph{minimum cardinality generator} problem for general closure systems, described in Section \ref{MinCaS}. In \cite{LuOs78}, C.L. Lucchesi and S.L. Osborn describe a problem equivalent to the minimum cardinality generator problem and prove NP-completeness by a reduction from the graph vertex cover problem. Section \ref{MinCaS} gives a proof of the intractability of the minimum cardinality generator problem by a reduction from the Boolean CNF satisfiability problem. Section \ref{OptBaS} completes the proof of intractability of optimum bases of convex geometries.

In Section \ref{DecConS}, we show that even the problem of deciding whether an implication basis defines a convex geometry is co-NP-complete. The proof is a direct reduction from the Boolean DNF tautology problem.

\section{Definitions and Background Theorems}

We use definitions and notational conventions from \cite{AdNa16I,CaMj03}.
A closure system $\langle X, \phi \rangle$ is a set, $X$, equipped with a function, $\phi: \mathcal{P}(X) \mapsto \mathcal{P}(X)$, satisfying the properties $A \subseteq \phi(A)$, $A \subseteq B \implies \phi(A) \subseteq \phi(B)$, and $\phi(\phi(A))=\phi(A)$ for all $A, B \subseteq X$; such a function $\phi$ is called a \emph{closure operator} on $X$, and the sets in the image of $\phi$ are called \emph{closed sets}. Closure systems can be equivalently defined in terms of their closed sets: if $\mathcal{F}$ is a family of subsets of $X$ such that $\mathcal{F}$ is closed under set intersection and $\mathcal{F}$ contains $X$, the function $\phi$ defined by $\phi(A) = \cap \{S \in \mathcal{F}: A \subseteq S \}$ for all $A \subseteq X$ is a closure operator on $X$, and the image of $\phi$ is precisely $\mathcal{F}$. Conversely, if $\langle X, \phi \rangle$ is a closure system and $\mathcal{F}$ is the image of $\phi$, then $\mathcal{F}$ is closed under set intersection, $\mathcal{F}$ contains $X$, and $\phi(A) = \cap \{S \in \mathcal{F}: A \subseteq S \}$ for all $A \subseteq X$. A closure system $\langle X, \phi \rangle$ is finite if its ground set, $X$, is finite. A closure system $\langle X, \phi \rangle$ is \emph{zero-closed} if the empty set is closed in the system, i.e. if $\phi(\emptyset)=\emptyset$.
\newline

A \emph{convex geometry} is a finite zero-closed closure system satisfying the anti-exchange property. For a finite closure system $\langle X, \phi \rangle$, the \emph{anti-exchange} property is the condition that $y \not \in \phi(A \cup \{x\})$ or $x \not \in \phi(A \cup \{y\})$ for every closed set $A \subsetneq X$ and all $x, y \in X \setminus A$ with $x \neq y$. A condition equivalent to the anti-exchange property for finite closure systems $\langle X, \phi \rangle$ is that for every closed set $A \subsetneq X$, there is an element $x \in X \setminus A$ such that $A \cup \{x\}$ is closed.
\newline

An \emph{implication} on a nonempty set $X$ is an ordered pair of sets $A, B \subseteq X$ where $B \neq \emptyset$, denoted $A \rightarrow B$. An implication $A \rightarrow B$ on $X$ holds in the closure system $\langle X, \phi \rangle$ if, for every closed set $S$ of $\langle X, \phi \rangle$, $A \subseteq S \implies B \subseteq S$, or equivalently if $B \subseteq \phi(A)$. For any set $\Sigma$ of implications on a set $X$, the family of sets $\mathcal{F} = \{ S \subseteq X: \forall (A \rightarrow B) \in \Sigma$: $A \subseteq S \implies B \subseteq S \}$ is closed under set intersection and contains $X$, and $\Sigma$ is called an \emph{implication basis} for the closure system on $X$ associated with $\mathcal{F}$. Given an implication basis $\Sigma$ of a finite closure system $\langle X, \phi \rangle$ and a set $S \subseteq X$, there are straightforward algorithms to compute $\phi(S)$.
\newline

A \emph{quasi-closed} set $Q$ of a closure system $\langle X, \phi \rangle$ is a subset of $X$ such that $Q$ is not closed and for every closed set $S$ of $\langle X, \phi \rangle$, $Q \subseteq S$ or $S \cap Q$ is closed; an equivalent definition is that $Q \subseteq X$ is quasi-closed if $Q$ is not closed and for all $S \subseteq Q$, $\phi(S)=\phi(Q)$ or $\phi(S) \subsetneq Q$. A \emph{critical set} $C$ of a closure system $\langle X, \phi \rangle$ is a quasi-closed set that is minimal among all quasi-closed sets of $\langle X, \phi \rangle$ that have the same closure as $C$; that is, a quasi-closed set $C$ is critical if for every quasi-closed set $Q \subsetneq C$, $\phi(Q) \subsetneq \phi(C)$. If $\langle X, \phi \rangle$ is a closure system, $\mathcal{F}$ is the family of all closed sets of $\langle X, \phi \rangle$, and $\mathcal{Q}$ is the family of all quasi-closed sets of $\langle X, \phi \rangle$, then the family of sets $\mathcal{F} \cup \mathcal{Q}$ is closed under set intersection and has an associated closure system, $\langle X, \sigma \rangle$, where $\sigma$ is called the \emph{saturation operator} of $\langle X, \phi \rangle$.

\begin{thm} \label{CanT}
If $\langle X, \phi \rangle$ is a finite closure system, $\sigma$ is the associated saturation operator, and $\Sigma$ is an associated implication basis, then for every critical set $C$ of $\langle X, \phi \rangle$, there is an implication $(A \rightarrow B) \in \Sigma$ such that $\sigma(A)=C$. Furthermore, for every finite closure system $\langle X, \phi \rangle$, the set $\Sigma_C = \{C \rightarrow \phi(C): C \text{ is a critical set of } \langle X, \phi \rangle\}$, called the \emph{canonical basis} of $\langle X, \phi \rangle$, is a valid implication basis of $\langle X, \phi \rangle$ and has minimal cardinality among all implication bases of $\langle X, \phi \rangle$.
\end{thm}
J.L. Guigues and V. Duquenne proved the statements of Theorem \ref{CanT} in \cite{GuDq86}, and D. Maier gave a similar result in the context of relational databases \cite{Mai80}. M. Wild later connected Guigues and Duquenne's work to Maier's and made their results significantly more accessible in \cite{Wild94}.
\newline

An implication basis $\Sigma_O$ of a finite closure system $\langle X, \phi \rangle$ is called an \emph{optimum basis} of $\langle X, \phi \rangle$ if the sum of all cardinalities of all left and right sides of the implications in $\Sigma_O$ is minimal among all implication bases for $\langle X, \phi \rangle$. It is possible to conclude from Theorem \ref{CanT} that optimum bases are shortenings of the canonical basis; that is, if $\Sigma_C$ is the canonical basis of a finite closure system and $\Sigma_O$ is an optimum basis of the same closure system, then there is bijection $f: \Sigma_C \mapsto \Sigma_O$ such that for all $(A \rightarrow B) \in \Sigma_C$, $f(A \rightarrow B) = C \rightarrow D$ where $C \subseteq A$ and $D \subseteq B$. \newline

The \emph{Boolean satisfiability problem} gives a Boolean proposition, $P$, on $n$ variables and asks if there exists a truth assignment, $\psi$, of the variables such that $P$ evaluates to \emph{true} according to $\psi$. The \emph{Boolean CNF satisfiability problem} is a special case of the Boolean satisfiability problem that only considers propositions written in \emph{conjunctive normal form}. A Boolean proposition is in \emph{conjunctive normal form} if it is expressed as a conjunction of clauses, and each clause is a disjunction of literals; a literal is either a single variable or the complement of a single variable. The Boolean CNF satisfiability problem is equivalent to the Boolean satisfiability problem in the sense that both are NP-complete. \newline

The \emph{Boolean tautology problem} is a dual problem to the Boolean satisfiability problem: the Boolean tautology problem gives a proposition, $P$, on $n$ variables and asks if there is a truth assignment, $\psi$, of the variables such that $P$ evaluates to \emph{false} according to $\psi$. The \emph{Boolean DNF tautology problem} is a special case of the Boolean tautology problem that only considers propositions written in \emph{disjunctive normal form}; a Boolean proposition is in \emph{disjunctive normal form} if it is expressed as a disjunction of clauses, and each clause is a conjunction of literals. The Boolean DNF tautology problem is dual to the Boolean CNF satisfyability problem, and both the Boolean DNF tautology problem and the general Boolean tautology problem are co-NP-complete.

\section{The Minimum Cardinality Generator Problem} \label{MinCaS}

A \emph{generator} of a closure system $\langle X, \phi \rangle$ is a set $S$ such that $\phi(S)=X$. Given a set of implications $\Sigma$, the \emph{minimum cardinality generator} problem asks to find a generator of minimal cardinality for the closure system associated with $\Sigma$. In \cite{LuOs78}, C.L. Lucchesi and S.L. Osborn studied an equivalent problem phrased in the language of functional dependencies for relational databases, and proved a result equivalent to Lemma \ref{MinCaL} below.

\begin{lem}\label{MinCaL}
The minimum cardinality generator problem is NP-hard.
\end{lem}
\begin{proof}
The proof proceeds by reduction from the Boolean CNF satisfiability problem. Let $P$ be a Boolean expression on $n$ variables, $a_1,...,a_n$, and suppose that $P$ is written in conjunctive normal form with clauses $c_1, ..., c_m$; that is, $P=c_1 \land ... \land c_m$, where each $c_j$ is a disjunction of literals.

Let $X = \{x_1,...,x_n,y_1,...,y_n,z_1,...z_m,s\}$, and define the set of implications $\Sigma$ as follows:
\begin{itemize}
\item For each clause $c_j$ of $P$ and each literal $a_i$ in $c_j$, $(x_i \rightarrow z_j) \in \Sigma$.
\item For each clause $c_j$ $P$ and each literal $\neg a_i$ in $c_j$, $(y_i \rightarrow z_j) \in \Sigma$.
\item $(z_1...z_m \rightarrow s) \in \Sigma$.
\item For each $i$, $(s x_i \rightarrow y_i) \in \Sigma$ and $(s y_i \rightarrow x_i) \in \Sigma$.
\end{itemize}

\noindent
Any generator of the closure system associated with $\Sigma$ must contain $x_i$ or $y_i$ for each $i$, and therefore must have cardinality at least $n$. On the other hand, the closure system  associated with $\Sigma$ has a generator of cardinality $n$ if and only $P$ is satisfiable.

Given a truth assignment $\psi$ of $(a_1,...,a_n)$, let $S \subseteq \{x_1,...,x_n,y_1,...,y_n\}$ be the set of cardinality $n$ such that $x_i \in S$ if (and only if) $a_i$ is true according to $\psi$ and $y_i \in S$ if (and only if) $a_i$ is false according to $\psi$. Then $z_j \in \phi(S)$ if and only if $c_j$ evaluates to true under $\psi$, and $s \in \phi(S)$ if and only if $P$ evaluates to true under $\psi$. It follows that $\phi(S)=X$ if and only if $P$ evaluates to true under $\psi$.
\end{proof}

\section{Optimum Bases of Convex Geometries}\label{OptBaS}

\begin{thm}\label{OptBaT}
The problem of finding an optimum basis for a convex geometry given an arbitrary implication basis is NP-hard.
\end{thm}

\begin{proof}
Let $n$ and $m$ be positive integers and let $X=\{x_1,...,x_m\}$. For each positive integer $i \leq n$, let $A_i$ and $B_i$ be nonempty subsets of $X$ such that $A_i \cap B_i = \emptyset$. Let $\Sigma = \{A_i \rightarrow B_i: 1 \leq i \leq n\}$ be an implication basis. Let $\phi$ be the closure operator associated with $\Sigma$. By Lemma \ref{MinCaL}, it suffices to reduce the problem of finding a minimum cardinality generator for the closure system associated with $\Sigma$ to the problem of optimizing an implication basis for a convex geometry.

For each $(A_i \rightarrow B_i) \in \Sigma$, let $\Sigma_i$ be the basis consisting solely of the implication $A_i \rightarrow B_i$. Let $Y=\{y_1,...,y_n\}$. For each $i$, let $\Sigma_i^{'}=\{(A_i \cup Y \setminus \{y_i\}) \rightarrow B_i\}$ be a singleton implication basis on the set $X \cup Y$. Let $\Sigma^{'}=(\cup_{i=1}^n \Sigma_i^{'}) \cup \{Y \rightarrow X\}$.
\begin{clm}\label{OptBaTClm1}
The closure system associated with $\Sigma^{'}$ is a convex geometry.
\end{clm}
\begin{clm}\label{OptBaTClm2}
Every optimum basis for the closure system associated with $\Sigma^{'}$ is of the form $(\cup_{i=1}^n \Sigma_i^{'}) \cup \{Y \rightarrow S\}$, where $S$ is a minimum cardinality generator for the closure system associated with $\Sigma$.
\end{clm}

\begin{proof}[Proof of Claim \ref{OptBaTClm1}]
Let $C$ be a closed set in the closure system associated with $\Sigma^{'}$. If $|C \cap Y| \leq n-2$, then $C \cup \{x\}$ is closed for every $x \in X \setminus C$; if $|C \cap Y| \leq n-2$ and $X \setminus C = \emptyset$, then $C \cup \{y\}$ is closed for every $y \in Y \setminus C$. Furthermore, if $|C \cap Y| = n = |Y|$, then $C$ must equal the top element, $X \cup Y$. So the only case left to consider is $|C \cap Y| = n-1$. Let $y_i$ be the sole element of $Y \setminus C$ and let $A_i, B_i$ be the subsets of $X$ with $((A_i \cup Y \setminus \{y_i\}) \rightarrow B_i) \in \Sigma^{'}$. Then for each $x \in X \setminus A_i$, $C \cup \{x\}$ is closed. If $X \setminus A_i \subseteq C$, then $B_i \subseteq C$, which implies that $C \cup \{x\}$ is closed for every $x \in X$. Lastly, if $X \subseteq C$, then $C \cup \{y_i\}$ is the top element (which is closed).
\end{proof}

\begin{proof}[Proof of Claim \ref{OptBaTClm2}]
Let $\phi^{'}$ be the closure operator associated with $\Sigma^{'}$. For each positive integer $i$ less than or equal to $n$, let $L_i=A_i \cup Y \setminus \{y_i\}$ and let $R_i=B_i$. Let $L_0=Y$ and let $R_0=X$. For each $i$, $\phi(L_i)=L_i \cup R_i$: when $i = 0$, $L_i \cup R_i$ is closed because it equals the top element, and when $i \neq 0$, $L_i \cup R_i$ is closed because $y_i \not \in L_i \cup R_i$ (since $L_i \cup R_i = (Y \setminus \{y_i\}) \cup (A_i \cup B_i)$ where $A_i \cup B_i \subseteq X$) and $y_i \in L_j$ for each $j \neq i$. 

Since all of the $A_i$ are nonempty, every proper subset of every $L_i$ is closed; but since $R_i \not \subseteq L_i$ for any $i$, none of the $L_i$ are closed. Also, the closures of the $L_i$ are distinct, since they have distinct intersections with $Y$. So for each $i$, $L_i$ is a critical set and is not the saturation of any set below it. Therefore, by Theorem \ref{CanT}, an optimum basis for the closure system associated with $\Sigma^{'}$ must consist of implications of the form $L_i \rightarrow S_i$ where $S_i \subseteq R_i$ for each $i$; furthermore, when $i$ is positive, $L_j \not \subseteq L_i \cup R_i$ for any $j \neq i$, so $S_i$ must equal $R_i$. 

Let $S$ be a subset of $X$ and let $\Sigma^{'}_O = \{L_i \rightarrow R_i: 1 \leq i \leq n\} \cup \{Y \rightarrow S\} = (\cup_{i=1}^n \Sigma_i^{'}) \cup \{Y \rightarrow S\}$. Let $\phi^{'}_O$ be the closure operator for the closure system associated with $\Sigma^{'}_O$. $\Sigma^{'}_O$ is a valid basis for the closure system associated with $\Sigma^{'}$ if and only if the implication $Y \rightarrow X$ can be derived in $\Sigma^{'}_O$, or equivalently, if and only if $X \subseteq \phi^{'}_O(Y)$. Since $L_i \setminus Y = A_i$ and $R_i=B_i$ for each positive $i$, $\phi^{'}_O(Y) = \phi^{'}_O(Y \cup S) = Y \cup \phi(S)$. So $\Sigma^{'}_O$ is a valid basis for the closure system associated with $\Sigma^{'}$ if and only if $S$ is a generator of the closure system associated with $\Sigma$. It follows that $\Sigma^{'}_O$ is an optimum basis for the closure system associated with $\Sigma^{'}$ if and only if $S$ is a minimum cardinality generator of the closure system associated with $\Sigma$.
\end{proof}
\let \qed \relax
\end{proof}

\section{Deciding if a Closure System is a Convex Geometry}\label{DecConS}

\begin{thm}\label{DecConT}
Given an arbitrary implication basis, the problem of determining whether the closure system associated with that basis is a convex geometry is co-NP-complete.
\end{thm}

\begin{proof}
Inclusion in co-NP is clear, because it is easy to validate a counterexample to the anti-exchange property when an implication basis is given (and it is easy to determine if the associated closure system is zero-closed). 

The proof of co-NP-hardness proceeds by reduction from the Boolean DNF tautology problem. Let $P$ be a Boolean expression on $n \geq 2$ variables, $a_1,...,a_n$, and suppose that $P$ is written in disjunctive normal form. Let $c_1, ..., c_m$ be the clauses of $P$; so $P=c_1 \lor ... \lor c_m$, where each $c_j$ is a conjunction of literals.

Let $X=\{x_1,...,x_n,t_1,...,t_n,f_1,...,f_n\}$. Let $x_{n+1}=x_1$. Define an implication basis $\Sigma$ on $X$ as follows:
\begin{itemize}
\item For each positive integer $i \leq n$, $(x_i t_i \rightarrow x_{i+1}) \in \Sigma$ and $(x_i f_i \rightarrow x_{i+1}) \in \Sigma$.
\item For each $c_j$, define the set $S_j \subseteq \{t_1,...,t_n,f_1,...,f_n\}$ as follows: $t_i \in S_j$ if (and only if) $c_j$ contains the literal $a_i$, and $f_i \in S_j$ if (and only if) $c_j$ contains the literal $\neg a_i$. For all $S_j$, $(S_j \rightarrow x_1) \in \Sigma$.
\end{itemize}
The closure system associated with $\Sigma$ is a convex geometry if and only if $P$ is a tautology.

If $P$ is not a tautology, then there is truth assignment, $\psi$, of the $a_i$ such that $P$ is false under $\psi$. Define $A \subseteq \{t_1,...,t_n,f_1,...,f_n\}$ such that $t_i \in A$ if (and only if) $a_i$ is true according to $\psi$ and $f_i \in A$ if (and only if) $a_i$ is false according to $\psi$. Since $P$ is false under $\psi$, all of its clauses are false under $\psi$ as well. It follows that no implication in $\Sigma$ has a left hand side that is a subset of $A$, so $A$ is closed. But then the closure system associated with $\Sigma$ fails the anti-exchange property, since, for any $i$, the closure of $A \cup \{x_i\}$ contains all of $\{x_1,...,x_n\}$.

Now suppose that $P$ is a tautology. It remains to show that for every closed set $C$ other than the top element, there is an $e \in X \setminus C$ such that $C \cup \{e\}$ is closed. Let $A$ be a closed set such that for every $i$, $t_i \in A$ or $f_i \in A$. Then there is some implication $(S \rightarrow x_1) \in \Sigma$ with $S \subseteq A$, so in fact $\{x_1,...,x_n\} \subseteq A$. Then $A \cup \{e\}$ is closed for all $e \in X \setminus A$. Let $B$ be a closed set such that for some $i$, $t_i \not \in B$ and $f_i \not \in B$. Then there are three cases:
\begin{itemize}
\item If $x_i \not \in B$, then $B \cup \{x_i\}$ is closed.
\item If there is $j \leq n$ with $x_{j+1} \in B$ and $x_j \not \in B$, then $B \cup \{x_j\}$ is closed.
\item If $\{x_1,...,x_n\} \subseteq B$, then, as before, $B \cup \{e\}$ is closed for all $e \in X \setminus B$.
\end{itemize}
\end{proof}

\section*{Acknowledgement}
I am grateful to Kira Adaricheva for her valuable support and professional guidance.

\end{document}